\documentclass[pra,twocolumn]{revtex4}

\usepackage{lmodern}
\usepackage[utf8]{inputenc}
\usepackage[T1]{fontenc}

\usepackage{graphicx}
\usepackage{subfigure}
\graphicspath{{./Figures/}}

\usepackage{amsmath}
\usepackage{amssymb,amsthm}
\usepackage{bm}
\usepackage{bbm}
\usepackage{xcolor}

\newtheorem{theorem}{Theorem}
\newtheorem{lemma}{Lemma}

\renewcommand{\bm}{\mathbf}
\newcommand{\B}{\mathcal B_L}

\newcounter{mycount}

\newcommand{\id}{\mathbbm{1}}

\newcommand{\N}{\mathcal N_R}

\newcommand{\be}[1]{ \begin{equation} \begin{aligned}{#1} }    
\newcommand{\ee}{\end{aligned} \end{equation}}

\usepackage{braket}
\usepackage{multirow,bigdelim}

\begin{document}

\title{Composite fermion basis for M-component Bose gases}
\author{V. Skogvoll}
\author{O. Liab\o tr\o}

\affiliation{Department of Physics, University of Oslo, P.O. Box 1048 Blindern,
0316 Oslo, Norway}

\date{\today}

\begin{abstract}

The composite fermion (CF) formalism produces wave functions that are not always linearly independent. 
This is especially so in the low angular momentum regime in the lowest Landau level, where a subclass of CF states, known as simple states, gives a good description of the low energy spectrum.
For the two-component Bose gas, explicit bases avoiding the large number of redundant states have been found. 
We generalize one of these bases to the $M$-component Bose gas and prove its validity. 
We also show that the numbers of linearly independent simple states for different values of angular momentum are given by coefficients of $q$-multinomials.
\end{abstract}
\pacs{ }

\maketitle

%%%%%%%%%%%%%%%%
%%%%%%%%%%%%%%%%

\section{Introduction}
\label{sec:intro}

Rapidly rotating atomic gases and the associated quantum phenomena have been studied quite extensively, for review see e.g. \cite{viefers-review, cooper-review, fetter-review}.
Over the years, different groups have been able to engineer multi-component Bose-Einstein condensates.
Examples include using hyperfine states of $^{87}\textrm{Rb}$ \cite{hyperfinestates1,hyperfinestates2} to create two-component mixtures, 
the use of optical traps to create spin-1 Bose-Einstein condensates with $^{87}\textrm{Rb}$ \cite{RbSpinorCondensate} and $^{23}\textrm{Na}$ \cite{NaSpinorCondensate1,NaSpinorCondensate2,NaSpinorCondensate3} with relatively small spin interaction contributions, 
and 
the use of Feshbach resonances to make two-component Bose condensates in $^{87}\textrm{Rb}- {}^{39}\textrm{K} $ \cite{Feshbach1}, $^{87}\textrm{Rb}- {}^{41}\textrm{K}$ \cite{Feshbach2}, and $^{87}\textrm{Rb} - {}^{85}\textrm{Rb}$ \cite{Feshbach3}.
The focus of this paper is multi-species Bose gases in the lowest Landau level (LLL) at low angular momentum. 
This regime was realized for single species systems in a recent experiment \cite{topoOrderedMatter}.
Hope is that such experiments could also be run for multi-component Bose gases and a good theoretical understanding of such systems is therefore of great interest. 

The history of constructing trial wave functions, 
from the Laughlin wave function \cite{laughlin83} to the conceptualization of composite fermion (CF) formalism \cite{jain89,jainbook} and trial wave functions addressing non-Abelian quantum Hall states \cite{moore91, read99, ardonne-schoutens99}, has shown success and taught us much about the quantum structures of such systems. 
The adaptation and application of the CF formalism has also shown promise when applied to weakly interacting rotating multi-component Bose gases \cite{marius14,vidar17}.
 In particular the space of trial wave functions spanned by 'simple states' (see section \ref{sec:SimpleStates}) has been shown to give significant overlap with the low energy part of the spectrum for both homogeneous and inhomogeneous interaction. 
 
 Simple states result from a projection to the LLL, which generally leads to linear dependencies between the resulting wave functions.
The ratio of the number of \textit{a priori} non-zero simple states to linearly independent ones can quickly get very large, especially in the low angular momentum regime, e.g. $107/8$ for a mixture of $3+2+1$ bosons at $L=4$ \cite{vidar17}.
Recent developments \cite{meyer-lia16} have been able to shed light on these linear dependencies for two-component gases, even giving explicit bases \cite{lia-meyer17}.

In this article we present an explicit basis for the space of simple states for the $M$-component Bose gas and the proof of its validity. 
This basis generalizes the two-component basis of Ref. \cite{lia-meyer17}.
Additionally we show that the numbers of basis states are given by coefficients of $q$-multinomials.

%%%%%%%%%%%%%%%%
%%%%%%%%%%%%%%%%

\section{M-Component Rotating Bose Gases}
\label{sec:theory}

In this paper, $\alpha,\beta$ will label the different species of particles while $i,j$ will label particles of the same species.
There are $M$  boson species in the system and $N_\alpha$ particles of species $\alpha$.
There are $A = \sum_{\alpha=1}^M N_\alpha$ particles in total.
Whenever we iterate over all particles, we will use indices $\mu,\nu$.
The particles of the first species have indices $\mu = 1,..,N_1$, the second $\mu = N_1+1,...,N_1+N_2$ and so on.
In general we can relate the species $\alpha$ and the particle number $i$ within that species to its index $\mu$ by 

\be{}
(\alpha,i) \leftrightarrow \mu = \left (\sum_{\beta=1}^{\alpha-1}N_\beta \right )  + i.
\ee

The Hamiltonian of an $M$-component Bose gas in a harmonic oscillator trap of strength $\omega$ under rotation with angular velocity $\Omega$ is given by  

\begin{multline}
\hat H =\sum_{\mu=1}^A \left ( \frac{\vec p_\mu^2}{2m} + \frac{1}{2} m \omega^2 \vec r_\mu^2 - \Omega \hat l_\mu \right )
\\
+
2\pi g 
\sum_{\mu=1}^A\sum_{\nu=\mu+1}^A
\delta(\vec r_\mu - \vec r_\nu).
\end{multline}
In the weak interaction limit, this reduces to the Lowest Landau level problem \cite{viefers-review} in the effective magnetic field $2 m \omega$ and the Hamiltonian is given by 

\be{}
\hat H =
\hat H_0 + \hat W
= \sum_{\mu = 1}^A
(\omega - \Omega) \hat l_\mu
+
\sum_{\mu < \nu=1}^A 2\pi g \delta(\vec r_\mu- \vec r_\nu).
\label{eq:Hamiltonian}
\ee
In the ideal limit $(\omega-\Omega) \rightarrow 0$, the Landau levels flatten and the eigenstates are solely determined by the interaction $\hat W$.

Single-particle eigenstates in the lowest Landau level with angular momentum $l$ are given by 

\be{}
\psi_{0,l}(z_\mu) = N_l z_\mu^l \exp(-z_\mu \bar z_\mu/4),
\ee
where $z_\mu = x_\mu + i y_\mu$ is the complex position variable of particle $\mu$ in units of the magnetic length $\sqrt{\hbar / (2m\omega)}$.
The simultaneous many-body eigenstates of $\hat H_0$ and the total angular momentum $\hat L = \sum_\mu \hat l_\mu$, with eigenvalues $L \hbar (\omega-\Omega)$ and $L$ respectively, are 
products of a ubiquitous Gaussian factor, $\exp\left (-\sum_{\mu=1}^A z_\mu \bar z_\mu/4 \right )$, and a homogeneous polynomial $P$ of degree $L$, symmetric in variables of the same species.

\section{Simple States}
\label{sec:SimpleStates}

A CF trial wave function for the bosonic $M$-component system is easily generalizable from the two-component case \cite{jainbook} and is on the form 

\be{}
\Psi_{CF} = \mathcal P_{LLL}
\left [
 \left (\prod_{\alpha=1}^M \Phi_\alpha\right ) \mathcal J^q
 \right ]
\exp \left ( -\sum_{\mu=1}^A \frac{z_\mu \bar z_\mu}{4} \right )
 ,
\ee
where $\mathcal P_{LLL}$ denotes the projection to the LLL and $\Phi_\alpha$ is a Slater determinant for the particles of species $\alpha$ and consists of orbitals from degenerate Landau-like levels  called $\Lambda$-levels \cite{lia-meyer17}.
$\mathcal J$ is the \textit{Jastrow factor} given by

\be{}
\mathcal J = \prod_{\mu<\nu} (z_\mu - z_\nu)
=
\sum_{\tau \in S_A} (-1)^{|\tau|}
\prod_{\mu = 1}^A
z_{\tau(\mu)}^{\mu - 1},
\ee
and $q$ is an \textit{odd} number to ensure the correct overall symmetry of the CF wave function. 
In this paper, we limit our attention to $q=1$. 
\textit{Simple states} are CF trial wave functions where only the lowest angular momentum state of each $\Lambda$-level is available to each species. 
This minimizes the CF cyclotron energy for a given $L$ \cite{meyer-lia16}.
The Slater determinants then only consist of powers of $\bar z$, translating to powers of $\partial_z$ after projection to the LLL by the projection of Girvin and Jach \cite{girvin-jach84} (called method I in Ref. \cite{jainbook}).
The polynomial part of a simple state can consequently be represented by an array $\bm n \in \mathbb N^A$ where $n_\mu$ are occupation numbers which correspond to occupied $\Lambda$-levels for the different particles.
We name these polynomials ''simple polynomials'' and they take the form

\be{}
P(\bm n)
=
\sum_{
\sigma \in S_{\oplus N_\alpha}}
(-1)^{|\sigma|}
\prod_{\mu=1}^A \partial_{\mu}^{n_{\sigma(\mu)}}
\mathcal J
\label{eq:SimplePolynomial}
\ee
where 

\be{}
S_{\oplus N_\alpha}
=
\bigoplus_{\alpha=1}^M
S_{N_\alpha},
\ee
and $S_{N_\alpha}$ is the symmetric group of $N_\alpha$ elements.
The degree of $P(\bm n)$ and the angular momentum of the corresponding simple state is 

\be{}
 L =\frac{A(A-1)}{2}- \sum_{\mu=1}^A n_\mu    .
\ee  

By the anti-symmetries of Eq. (\ref{eq:SimplePolynomial}), a simple polynomial $P(\bm n)$ is anti-symmetric under the interchange of elements $n_\mu \leftrightarrow n_\nu$ when $\mu$ and $\nu $ belong to the same species.
This allows us to represent a non-zero simple polynomial pictorially (up to a sign) in a grid $\{1,2,...,M\}\times\{0,1,2,...,A-1\}$ where a $\circ$ at position $(\alpha,n)$ corresponds to some $n_{\alpha,i}=n$, since repeated values of $\bm n$ within a species is $0$ due to anti-symmetry. 
Figure \ref{fig:SSexample} shows an example of such a pictorial representation.

\begin{figure}
\centering 
	\includegraphics[scale=1]{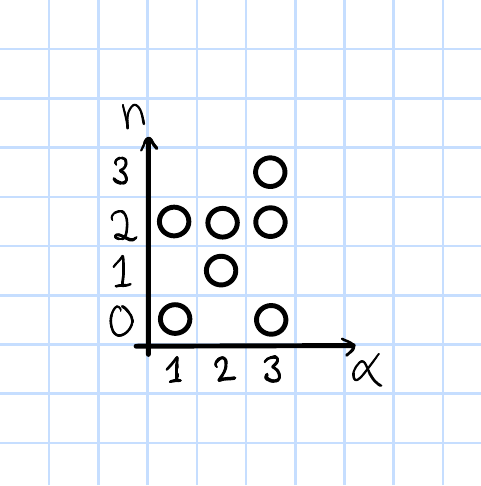}
	\caption{Pictorial representation of $P(\bm n)$ for $(N_1,N_2,N_3)=(2,2,3)$ given by $\bm n = (0,2;1,2;0,2,3)$.
	The semicolons in $\bm n$ indicate which occupation numbers belong to which species.}
	\label{fig:SSexample}
\end{figure}

As previously discovered \cite{marius14} for the two-component Bose gas and recently confirmed \cite{vidar17} for the $M$-component Bose gas, the space of simple states overlaps significantly with the low-energy part of the delta-potential interaction spectrum.
Figure \ref{fig:SimpleStateDiagonalization} shows an example of this, and perhaps most striking is the large overlap of the ground state with the space of simple states for the values of $L$ where the dimension of the latter is but a small fraction of the dimension of the relevant sector of the Hilbert space.
Seeing that the space of simple states has this characteristic, an explicit basis for it is of great interest.

\begin{figure*}
\hspace*{-1.1cm}
\subfigure[]{\includegraphics[scale=0.9]{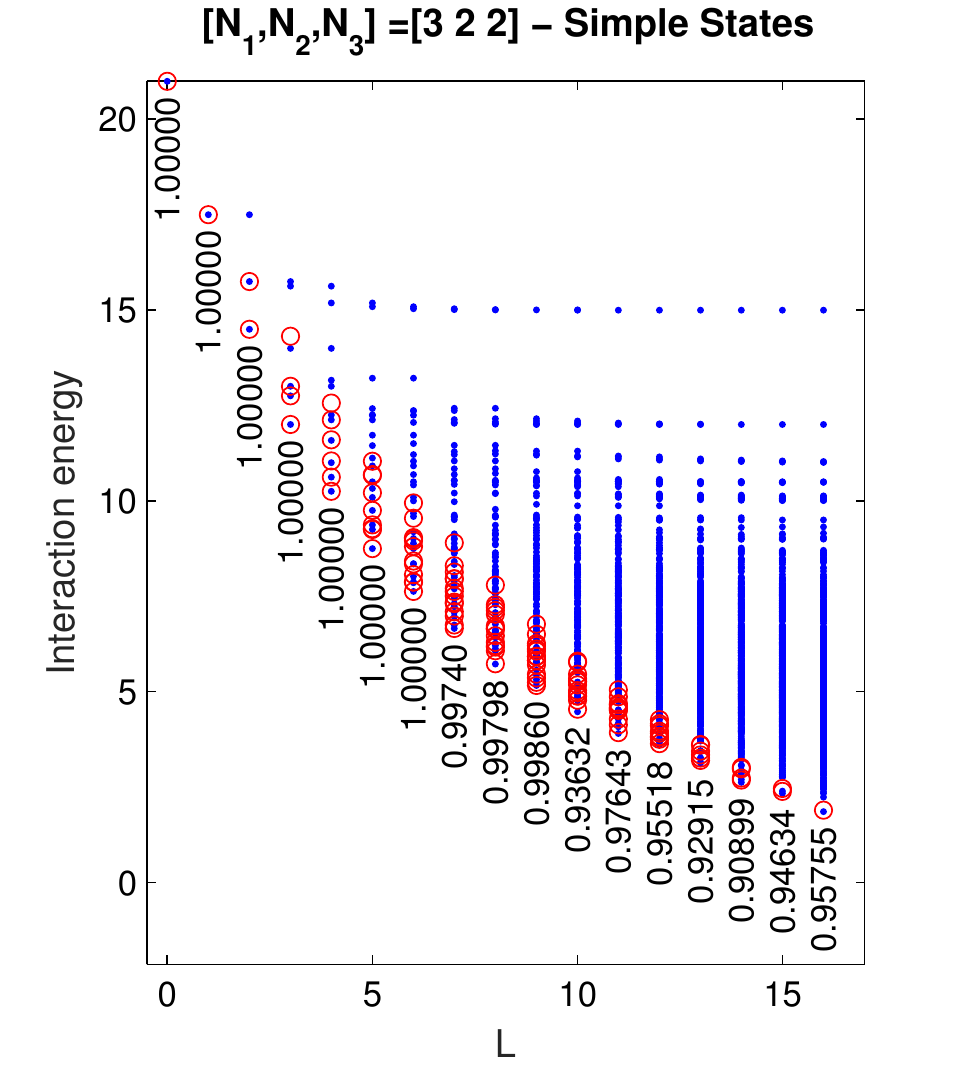}}
\subfigure[]{\includegraphics[scale=0.9]{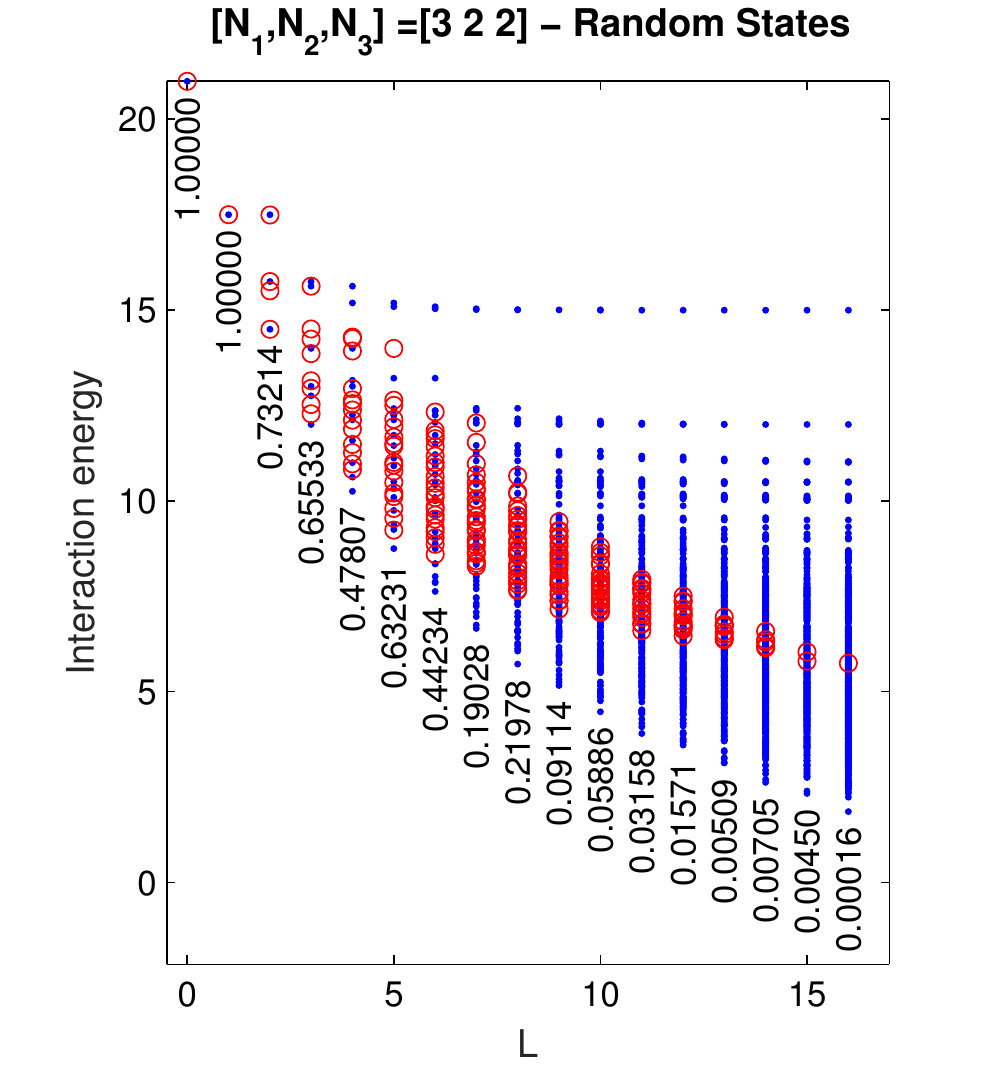}}
	\caption{The exact diagonalization (dots) of the homogeneous zero-range interaction potential within the Hilbert Space of translationally invariant states for $(N_1,N_2,N_3) = (3,3,2)$ together with the diagonalization (rings) within 
	(a) the subspace of simple states and 
	(b) a random subspace of equal dimension. 
	Numbers denote the squared absolute values of the ground state projection onto the subspaces.}
	\label{fig:SimpleStateDiagonalization}
\end{figure*}

We define the elementary differentiation polynomial of degree $R$ by 

\be{}
\hat d_R =  \N \sum_{\tau \in S_A} \prod_{\nu=1}^R \partial_{\tau(\nu)},
\ee
where $\N = \frac{1}{R!(A-R)!}$.
$\hat d_R$ is the sum over all unique products of $R$ out of $A$ first order partial derivatives.
The action of $\hat d_R$ on a simple polynomial is given by

\begin{multline}
\hat d_R  P(\bm n) =
\N \sum_{\tau \in S_A} \prod_{\mu=1}^R \partial_{\tau(\nu)}
\sum_{
\sigma \in S_{\oplus N_\alpha}}
(-1)^{|\sigma|}
\prod_{\mu=1}^A \partial_{\mu}^{n_{\sigma(\mu)}}
\mathcal J \\
=
\N \sum_{\tau \in S_A}
\sum_{
\sigma \in S_{\oplus N_\alpha}}
(-1)^{|\sigma|}
 \prod_{\nu=1}^R \partial_{\tau(\nu)}
\prod_{\mu=1}^A \partial_{\tau(\mu)}^{n_{\tau(\sigma(\mu))}}
\mathcal J \\
=
\N \sum_{\tau \in S_A}
\sum_{
\sigma \in S_{\oplus N_\alpha}}
(-1)^{|\sigma|}
 \prod_{\nu=1}^R \partial_{\tau(\nu)}^{n_{\tau(\sigma(\nu))}+1}
\prod_{\mu=R+1}^A \partial_{\tau(\mu)}^{n_{\tau(\sigma(\mu))}}
\mathcal J \\
=
\N \sum_{\tau \in S_A}
P(\bm n + \sum_{\nu=1}^R \bm e_{\tau(\nu)}),
\end{multline}
where $\bm e_\nu \in \mathbb N^A$ is a unit vector in the $\nu$'th direction. 
Lemma \ref{lemma:SymmetricTranslationalInvariance}  states that simple polynomials obey 
\be{}
\hat d_R P(\bm n) = 0.
\ee
This is an equivalent formulation of generalized translation invariance, as was used in \cite{meyer-lia16,lia-meyer17}.

\section{A basis for the simple states}

We now present our main result.

\begin{theorem}

The set $\B = \{P(\bm n)|  \text{ Eqs. }(\ref{eq:SSstepCritereon}) - (\ref{eq:SSangmomCritereon})\}$, where

\be{}
\label{eq:SSstepCritereon} n_{\alpha,i} < \sum_{\beta=1}^{\alpha} N_\beta,
\ee
\be{}
\label{eq:SSuniqueCritereon} n_{\alpha i} < n_{\alpha j} \Leftrightarrow i<j 
\ee
\be{}
\label{eq:SSangmomCritereon} \sum_{\mu=1}^A n_\mu =  \frac{A(A-1)}{2} - L
\ee

is a basis for the space of simple polynomials of degree $L$.
Consequently the corresponding set of simple states $\left \{ P(\bm n)  \exp\left (-\sum_{\mu=1}^A z_\mu \bar z_\mu/4 \right ) \right \}$ is a basis for the space of simple states with angular momentum $L$.
\end{theorem}

Pictorially, a simple polynomial in $\B$ corresponds to a figure where no $\circ$'s occupy spaces above the step function which increases by $N_\alpha$ for each step $\alpha-1\rightarrow \alpha$, see figure \ref{fig:BasisLimit}.
An example of two polynomials from the same space where one is not in $\B$ is given in figure \ref{fig:basisExample}.

\begin{figure}
\centering 
	\includegraphics[scale=1]{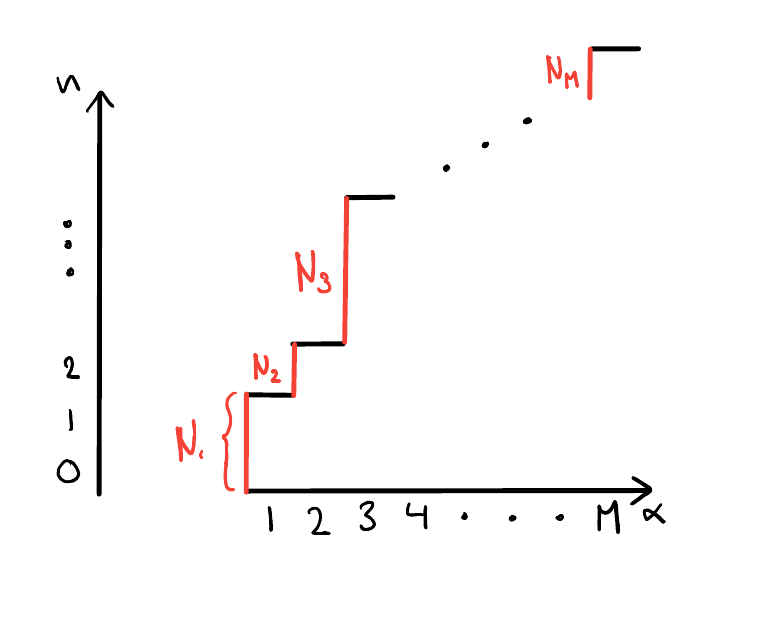}
	\caption{The step function that limits the configuration space of the basis polynomials when represented pictorially.}
	\label{fig:BasisLimit}
\end{figure}

\begin{figure}
\centering 
	\includegraphics[scale=0.8]{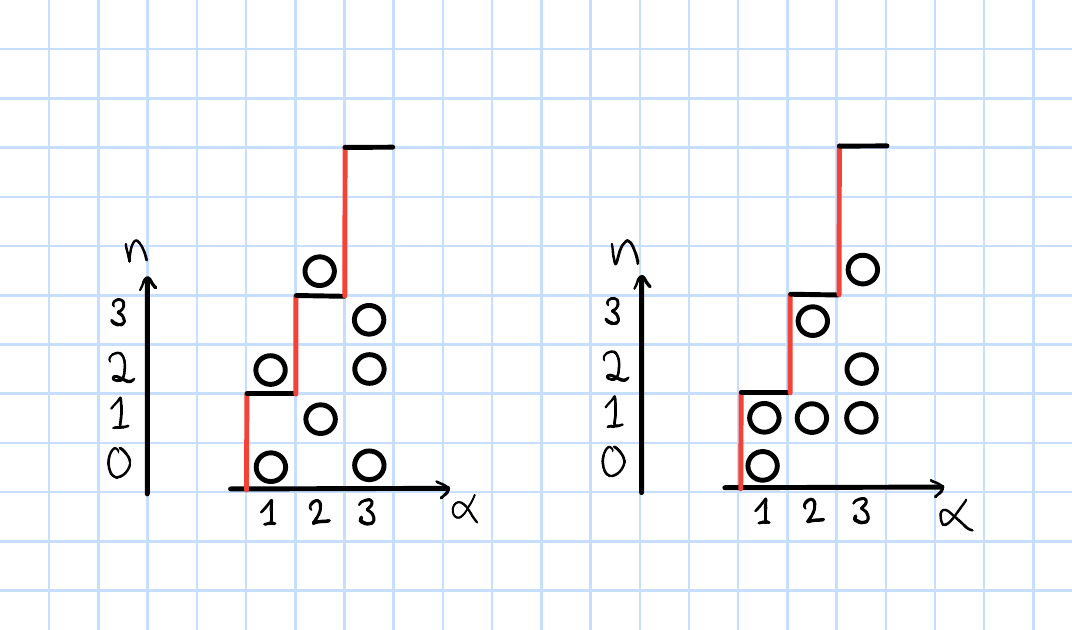}
	\caption{Example of two simple polynomials for the system $(N_1,N_2,N_3) = (2,2,3)$, both of degree $L=9$. 
	The left polynomial is not in $\B$ whereas the right one is. }
	\label{fig:basisExample}
\end{figure}

\begin{proof}

To prove that $\B$ is a basis we need to show that it spans the whole space of simple polynomials and that it is a linearly independent set.
We refer to Appendix \ref{sec:AppendixLemmas} for the Lemmas used in this section.

\subsubsection{Spanning of the space of simple polynomials}
Eq. (\ref{eq:SSuniqueCritereon}) picks a unique permutation within each species which fixes the otherwise ambiguous sign of the pictorial representation.
Eq. (\ref{eq:SSangmomCritereon}) ensures that the polynomials have the correct degree. 
The non-trivial result is that Eq. (\ref{eq:SSstepCritereon}) gives a basis under these conditions.

As shown in Lemma \ref{lemma:SSidentity} by induction, a consequence of Lemma \ref{lemma:SymmetricTranslationalInvariance} is that for all $R\leq n_\mu$

\be{}
P(\bm n)
=
(-1)^R
\hat d_R^{\sim \mu}
P(\bm n - R \bm e_\mu),
\ee
where

\be{}
\hat d_R^{\sim \mu}=
\N \sum_{\tau \in S_A} \prod_{\nu=1}^R \partial_{\tau(\nu)} (1-\delta_{\mu \tau(\nu)}).
\ee
The action of $\hat d_R^{\sim \mu}$ on a simple polynomial gives the sum of all simple polynomials where $R$ occupation numbers excluding the $\mu$'th are increased by $1$.
So Lemma \ref{lemma:SSidentity} means that we can write a simple polynomial $P(\bm n)$ as the sum of all simple polynomials where one, the $\mu$'th, occupation number of $\bm n$ has been reduced by $R$ units and $R$ other (not $ \mu$) have been increased by one.

We now introduce a partial ordering $\prec$ of the simple polynomials as follows:

\begin{multline}
P(\bm n') \prec P(\bm n)  \quad \textrm{ If for an } \alpha \\
\sum_{i} n'_{\alpha i} > \sum_{i} n_{\alpha i}
\quad \textrm{and} \quad \forall \beta > \alpha \left (\sum_{i} n'_{\beta i} = \sum_{i} n_{\beta i}   \right ).
\end{multline}
This means that a polynomial is ''smaller'' wrt. $\prec$ than another if the sum of its occupation numbers for the last species is bigger than the other. 
And if those are equal, the one with the largest occupation number sum for the next to last species is ''smaller'' and so on.
For a given total sum of occupation numbers in $\bm n$, there is a limit as to how ''small'' a state can get since $n_{\alpha i} \geq 0$.
Now consider a simple polynomial $P(\bm n^*)$ of degree $L$ which is not in $\B$ because it violates Eq. (\ref{eq:SSstepCritereon}).
This implies that it contains some exponent

\be{}
n^*_{\mu} = n^*_{\alpha, i} \geq \sum_{\beta \leq \alpha} N_\beta.
\label{eq:OffBasisOccNum}
\ee
We rewrite $P(\bm n^*)$ according to Lemma \ref{lemma:SSidentity}

\be{}
P(\bm n^*) = 
(-1)^{n^*_{\mu}}
\hat d_{n^*_{\mu}}^{\sim \mu}
P(\bm n^* - n^*_\mu \bm e_{\mu}).
\ee

The right hand side of this equation is a sum of polynomials where the $\mu$th occupation number is $n^*_\mu$ less than in $P(\bm n^*)$ and $n^*_\mu$ other occupation numbers are $1$ more than in $P(\bm n^*)$. 
The number of occupation numbers for species $\beta \leq \alpha$ which can be increased in this way is $\sum_{\beta \leq \alpha} N_\beta - 1$.
We see from Eq. (\ref{eq:OffBasisOccNum}) that this is less than $n^*_\mu$ and by the pigeon hole principle at least one occupation number must have increased for a species $\beta>\alpha$.
This means that we have rewritten $P(\bm n^*)$ in terms of polynomials that are ''smaller'' wrt. $\prec$. 
For any simple polynomial not in $\B$ because it violates Eq. (\ref{eq:SSstepCritereon}) we can perform this iteration and since there is a lower limit to how ''small'' a state can be, at some point we must have rewritten $P(\bm n^*)$ in terms of polynomials that obey Eq. (\ref{eq:SSstepCritereon}).

\subsubsection{Linear independence of the polynomials}
To prove the linear independence of the basis states we must return to the full form of a simple state given by $\bm n$.

\begin{multline}
P(\bm n)
=
\sum_{
\sigma \in S_{\oplus N_\alpha}}
\sum_{\tau \in S_A} (-1)^{|\tau|+|\sigma|}
\prod_{\mu=1}^A \partial_{\sigma(\mu)}^{n_{\mu}}
\prod_{\nu = 1}^A
z_{\tau(\nu)}^{\nu - 1} \\
=
\sum_{\tau \in S_A} (-1)^{|\tau|}
\sum_{
\sigma \in S_{\oplus N_\alpha}}
\prod_{\mu=1}^A \partial_{\sigma(\mu)}^{n_{\mu}}
z_{\sigma(\mu)}^{\tau(\mu) - 1}  \\
=
\sum_{\tau \in S_A}  t_{\bm n, \tau}.
\end{multline}
$t_{\bm n, \tau}$ is a symmetric polynomial within each species given by

\be{}
t_{\bm n,\tau} = 
(-1)^{|\tau|}
K_{\bm n, \tau}
\sum_{
\sigma \in S_{\oplus N_\alpha}}
\prod_{\mu=1}^A 
z_{\sigma(\mu)}^{\tau(\mu) - 1 - n_{\mu}} ,
\ee
where $K_{\bm n,\tau}$ is a constant obtained from differentiation, zero if $n_\mu > \tau(\mu)-1$ for any $\mu$.
We define an ordering of non-zero terms by

\be{}
t_{\bm n,\tau} < t_{\bm n',\tau'}
\ee 
if for the last particle species the $k$th least exponent of $t_{\bm n,\tau}$ is greater than the $k$th least exponent of $t_{\bm n',\tau'}$ and their $k-1$ least exponents are pairwise equal. 
If this is the case for all exponents for particles in the exponents in the last species, the terms are sorted according to the $M-1$ particle species and so on. 
Consider $t_{\bm n,\id}$ where $\id$ is the identity permutation

\be{}
t_{\bm n,\id}
=
K_{\bm n, \id}
\sum_{
\sigma \in S_{\oplus N_\alpha}}
\prod_{\mu=1}^A 
z_{\sigma(\mu)}^{\mu - 1 - n_{\mu}}.
\ee
This term is non-zero for all simple polynomials that obey Eq. (\ref{eq:SSstepCritereon}) while zero for simple polynomials that do not because for the latter at least \textit{one} $n_\mu$ is greater than $\mu-1$. 
For any other term $ t_{\bm n,\tau}$ in $P(\bm n)$

\be{}
t_{\bm n,\id} < t_{\bm n,\tau}.
\ee
Since every simple polynomial in $\B$ has a unique $t_{\bm n,\id}$,
 we can impose a complete ordering on the polynomials in $\B$ by 

\be{}
P(\bm n) <  P(\bm n') \quad \Leftrightarrow \quad t_{\bm n,\id} < t_{\bm n',\id}.
\ee
The smallest simple polynomial $P(\bm n_0)$ with respect to the ordering then has the property that $t_{\bm n_0,\id}$ does not occur in $P({\bm n})$ for any $P(\bm n) > P(\bm n_0)$. 
In general $t_{\bm n,\id}$ is not a term in $P(\bm n')$ for any $P(\bm n') > P(\bm n)$.
We numerate the simple states in $\B$ from $0$ to $T$ according to this relation.
Consider then a linear dependence relation with all $T$ simple states from the basis $\B$

\be{}
c_0 P(\bm n_0) + c_1 P(\bm n_1) + ... + c_T P(\bm n_T) = 0.
\ee
Since $t_{\bm n_0,\id}$ is in $P(\bm n_0)$, but not in $P(\bm n_i)$ for $i>0$, $c_0$ must be zero. 
By induction all coefficients must be $0$ and the states are linearly independent by definition.

\end{proof}

\section{On the number of states}
\label{sec:numberofSS}
We now address the problem of counting the simple polynomials in $\B$.
To do so we utilize a one-to-one correspondence between a simple polynomial and a certain class of multipartitions which we have called \textit{simple multipartitions} (SMPs). 
We refer to Appendix \ref{sec:AppendixSMPs} for the mathematical details of how to arrive at the result of this section.

%This number must be independent of the order of $N_1,...,N_M$, even though the basis depends on this order.  This translates to a combinatorial identity concerning the number of possible partitions of an integer.  We will also see that there is a closed form expression for the number of states in $\B$, which is explicitly independent of the order.

We use the $q$-analog of the factorial to present a closed form expression for the number of simple polynomials of degree $L$, which we denote $k_L(N_1,...,N_M)$.  
The $q$-analog of a positive integer is $[n]_q = \sum_{i=0}^{n-1} q^i$, and the $q$-factorial $[n]_q!$ is defined recursively by $[n]_q!=[n]_q [n-1]_q!$ with $[0]_q! = 1$.  
The numbers $k_L(N_1,...,N_M)$ then appear as the coefficients of the $q$-multinomial

\be{}
\sum_{L=0}^{L_{MAX}} k_L(N_1,...,N_M) q^L=\frac{[N_1 + ... + N_M]_q!}{[N_1]_q!\times...\times[N_M]_q!}
\label{eq:qmultinomial}
\ee
where

\be{}
L_{MAX} = \sum_{\alpha=1}^M \sum_{\beta=\alpha+1}^M N_\alpha N_\beta
\ee
is the largest possible degree of a simple polynomial.
We see that eq. (\ref{eq:qmultinomial}) reflects the fact that the number of states in $\B$ is independent of the order $N_1,...,N_M$, even though the actual states of $\B$ depend on said order.
It is instructive to see an example of how to use eq. (\ref{eq:qmultinomial}). 
Consider the system $(N_1,N_2,N_3) = (3,2,1)$.
The corresponding $q$-multinomial is given by 

\begin{multline}
\frac{[3+2+1]_q!}{[3]_q![2]_q![1]_q!}
=
\frac{[6]_q[5]_q[4]_q[3]_q!}{[3]_q![2]_q![1]_q!}
=
\frac{[6]_q[5]_q[4]_q}{[2]_q[1]_q[1]_q} \\
=
\frac{\begin{array}{l} (1+q+q^2+q^3+q^4+q^5) \\ \quad \times (1+q+q^2+q^3+q^4) \\ \qquad \times (1+q+q^2+q^3)\end{array}}{(1+q)(1)(1)} \\
= 1 + 2q + 4q^2 + 6q^3 + 8q^4 + 9q^5 + 9q^6 + 8q^7 + 6q^8 + 4q^9 + 2q^{10} + q^{11}.
\end{multline}
Which means, for instance, that the number of simple polynomials of degree $L=4$ is $8$.

\section{Summary}

We have found a basis for the space of $M$-component simple states generalizing one of the bases found for the two-component simple states \cite{lia-meyer17} along with a proof of its validity. 
Redundant states due to linear dependencies is a common issue with the CF formalism and our result gives a contribution to the understanding of these dependencies.
The more general problem of understanding the linear dependencies between all compact states is still open. 

The numbers of basis states for different values of angular momentum turn out to be given by coefficients of $q$-multinomials.
Linking these numbers to a certain class of multipartitions has lead to an interesting visual interpretation of the fact that the $q$-multinomials are independent of the order $N_1,...,N_M$
(see Appendix \ref{sec:AppendixSMPs} and figure \ref{fig:YoungDiagrams123}).

\section{Acknowledgement}

We thank our collaborator Marius L. Meyer for much help with the implementation of the various numerical methods used to discover and verify the correct basis generalization, and many a helpful discussion.
We also thank Susanne F. Viefers for many insightful comments on the written manuscript and Jørgen Vold Rennemo for pointing out the connection between SMPs and $q$-multinomials.
This research was financially supported by the Research Council of Norway.

\appendix

\section{Simple multipartitions (SMPs) and $q$-multinomials}
\label{sec:AppendixSMPs}

We define an SMP as a multipartition of an integer into $(M-1)$ partitions in rectangles (i.e. with limits on both the maximal, and the number of non-zero, elements). 
The dimensions of the rectangles are given by $M$ numbers $N_1,...,N_M$ such that the dimensions of rectangle $\alpha$ are given by 

\be{}
\left (\sum_{\beta\leq \alpha}N_\alpha \right) \times N_{\alpha+1}.
\ee
We can represent an SMP of an integer $L$ by a vector $\bm p \in \mathbb N^A$ 

\be{}
\bm p = [p_{1,1},p_{1,2},...,p_{M,N_M}]
\ee
where $\sum_{\alpha,i} p_{\alpha,i} = L$, $p_{\alpha,i}\leq p_{\alpha,j}$ if $i>j$ and $p_{\alpha,i} \leq \sum_{\beta=1}^{\alpha-1} N_\beta$.
The possible values of $L$ are within $\{0,...,\sum_{\beta<\alpha=1}^M N_\alpha N_\beta\}$.
An SMP corresponds to compactly coloring $L$ cells to make $M-1$ Young diagrams in rectangles of dimensions $N_1 \times N_2$, $(N_1+N_2) \times N_3$,..., $\left (\sum_{\beta=1}^{M-1} N_\beta\right ) \times N_M$ as shown in figure \ref{fig:RPCP}.

\begin{figure}
\centering
\includegraphics[scale=0.6]{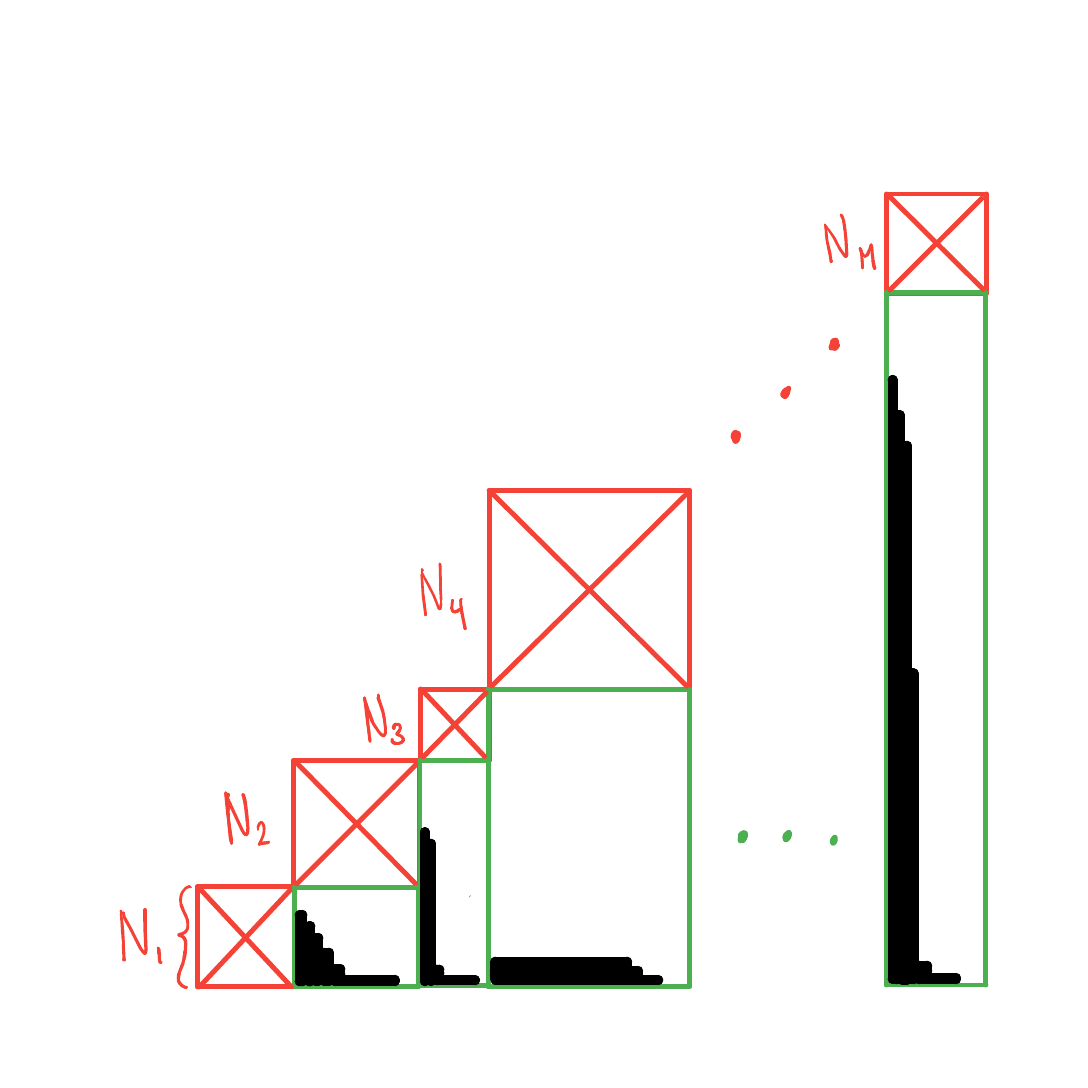}
\caption{The compact coloring of $L$  cells to make $M-1$ Young diagrams in rectangles of dimensions $N_1 \times N_2$, $(N_1+N_2) \times N_3$,..., $\left (\sum_{\beta=1}^{M-1} N_\beta\right ) \times N_M$.
Such a coloring corresponds uniquely to an SMP of $L$ where the height of column $i$ of the rectangle with width $N_\alpha$ is $p_{\alpha,i}$.}
\label{fig:RPCP}
\end{figure}

The one-to-one correspondence between a simple polynomial $P(\bm n) \in \B$ of degree $L$  and an SMP of $L$ is given by

\be{}
p_{\alpha,i} = \left (\sum_{\beta=1}^{\alpha-1} N_\beta \right ) + (i-1) - n_{\alpha,i}.
\ee
Note that since the number of basis states in $\B$ is independent of the order $N_1,...,N_M$, this one-to-one correspondence imply that the number of possible SMPs of an integer $L$ is independent of the order $N_1,...,N_M$. 
This is a generalization of the trivial fact that the number of possible compact colorings of $L$ cells in an $N_1\times N_2$ Young Diagram is the same as in an $N_2 \times N_1$ Young Diagram.
An example of this with $L=4$, $(N_1,N_2,N_3) = (1,2,3)$ is given in figure \ref{fig:YoungDiagrams123}.

\begin{figure*}
\centering
\includegraphics[width=\textwidth]{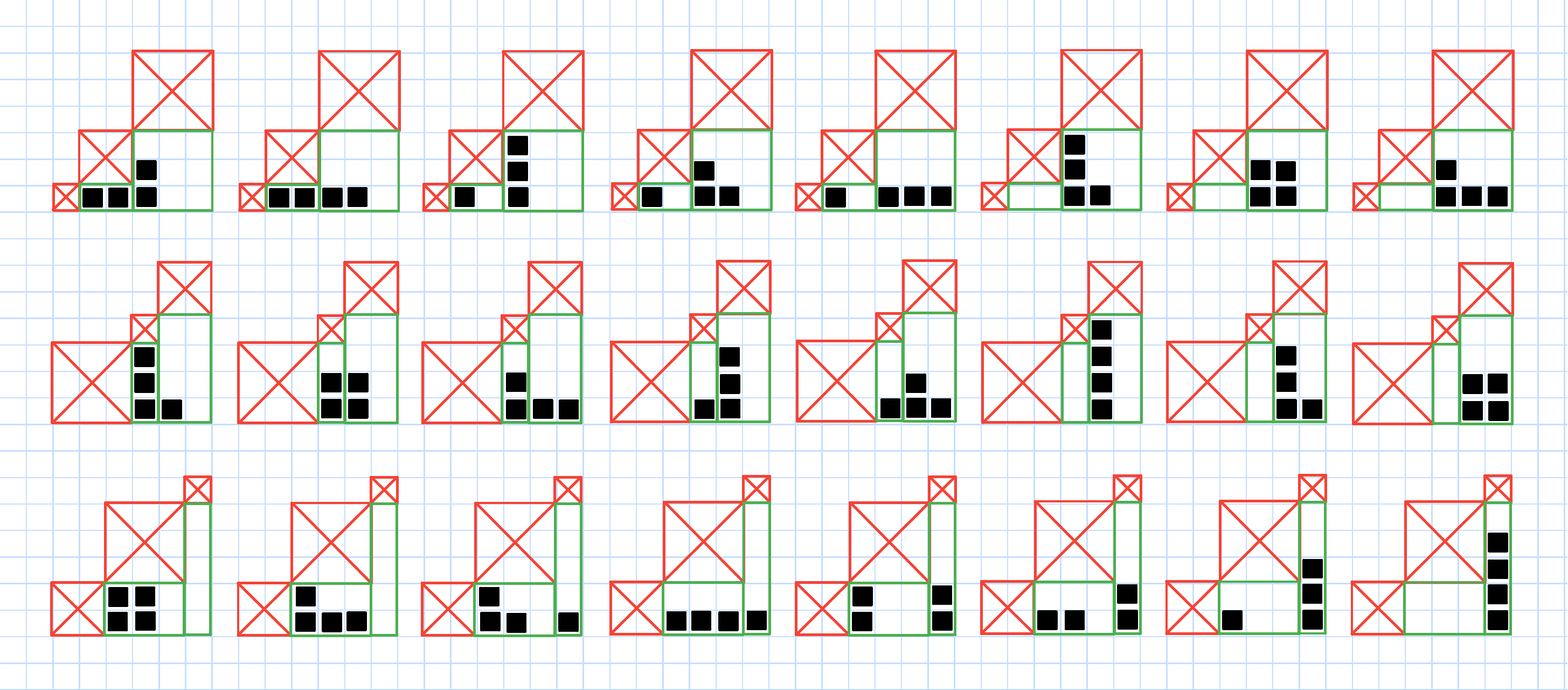}
\caption{Example of the compact coloring of $L$ cells in different combined Young diagrams for $(N_1,N_2,N_3)=(1,2,3)$ (first row), $(3,1,2)$ (second row) and $(2,3,1)$. 
The number of such colorings is independent of permutations of $(N_1,N_2,N_3)$.}
\label{fig:YoungDiagrams123}
\end{figure*}

The number of SMPs of $L$ is related to $q$-multinomials (see e.g. Ref. \cite{qmultinomials} for information on these).
A $q$-binomial is defined as follows

\be{}
\begin{bmatrix}
n\\
k
\end{bmatrix}_q
=
\frac{[n]_q!}{[n-k]_q![k]_q!}
\ee
where $[n]_q$ and $[n]_q!$ have been defined in section \ref{sec:numberofSS}.
The number $k_L(N_1,N_2)$ of compact colorings of $L$ cells in an $N_1\times N_2$ Young Diagram is related to $q$-binomials by (see e.g. Ref. \cite{MITopencourseware})

\be{}
\sum_{L=0}^{N_1 N_2} k_L(N_1,N_2)q^{L} = \begin{bmatrix} N_1  + N_2 \\ N_1 \end{bmatrix}_q
=
\frac{[N_1 + N_2]_q!}{[N_1]_q![N_2]_q!}.
\label{eq:qbinomYoungLink}
\ee
The product 

\be{}
\left (
\sum_{l=0}^{N_1 N_2} k_l(N_1,N_2)q^{l}
\right )
\left (
\sum_{l'=0}^{(N_1 +N_2) N_3} k_{l'}(N_1+N_2,N_3)q^{l'}
\right )
\ee
gives a polynomial in $q$ whose coefficient before $q^L$ is the number of ways to compactly color $L=l+l'$ cells in two Young Diagrams of dimensions $N_1\times N_2$ and $(N_1+N_2)\times N_3$.
We see from the right hand side of eq. (\ref{eq:qbinomYoungLink}) that this product equals 

\begin{multline}
\begin{bmatrix} N_1  + N_2 \\ N_1 \end{bmatrix}_q
\begin{bmatrix} N_1  + N_2 + N_3 \\ N_1 + N_2 \end{bmatrix}_q\\
=
\frac{[N_1 + N_2]_q!}{[N_1]_q![N_2]_q!}
\frac{[N_1 + N_2 + N_3]_q!}{[N_1+N_2]_q![N_3]_q!} 
=
\frac{[N_1 + N_2 + N_3]_q!}{[N_1]_q![N_2]_q![N_3]_q!} 
\end{multline}
where we identify the last expression as the $q$-multinomial. 
We have thus shown eq. (\ref{eq:qmultinomial}) for the case of $3$ species. 
The generalization to $M$ species is straightforward.

\section{Lemmata}
\label{sec:AppendixLemmas}

\begin{lemma}[Symmetric Translation Invariance]
\label{lemma:SymmetricTranslationalInvariance}
\be{}
\hat d_R P(\bm n) = 0
\ee
For $R\geq 1$.
\end{lemma}

\begin{proof}

$\hat d_R$ consists only of differential operators and commutes with the differential operators in $P(\bm n)$.
We therefore only need to show that $\hat d_R \mathcal J = 0$.

\begin{multline}\label{lemma1dRJ}
\hat d_R \mathcal J
=
\N \sum_{\tau \in S_A} \prod_{\mu=1}^R \partial_{\tau(\mu)}
\sum_{\sigma \in S_A} (-1)^{|\sigma|}
\prod_{\nu = 1}^A
z_{\sigma(\nu)}^{\nu - 1} \\
=
\N \sum_{\tau \in S_A} 
\sum_{\sigma \in S_A} (-1)^{|\sigma|}
\prod_{\mu=1}^R \partial_{\tau(\mu)}
\prod_{\nu = 1}^A
z_{\sigma(\nu)}^{\nu - 1} \\.
\end{multline}

Consider a term $t_1$ in this sum given by $\tau=\tau_1$ and $\sigma=\sigma_1$ and consider the lowest value $\nu'$ for which $\sigma_1(\nu')=\tau_1(\mu_1)$ for some $\mu_1 \in \{1,2,...,R\}$.  If $\nu'=1$, then $t_1$ is zero due to differentiation of the constant $z_{\sigma_1(1)}^0$.  Else, if $\nu'>1$, we may consider the factor $z_{\sigma_1(\nu'-1)}^{\nu'-2}$ and write   
\be{}\label{lemma1t1}
t_1=(-1)^{|\sigma_1|}z_{\sigma_1(\nu'-1)}^{\nu'-2}\partial_{\tau_1(\mu_1)}z_{\sigma_1(\nu')}^{\nu'-1}\times\dots
\ee
where we have written out only the permutation sign and the variables and differentiation operators of type $\sigma_1(\nu')$ and $\sigma_1(\nu'-1)$.  We can now pick another term, $t_2$, given by permutations $\sigma_2$ and $\tau_2$.  The permutation $\sigma_2$ is the same as $\sigma_1$ except that 
$\sigma_2(\nu')=\sigma_1(\nu'-1)$ and $\sigma_2(\nu'-1)=\sigma_1(\nu')$.  We find the $\mu_2>R$ such that $\tau_1(\mu_2)=\sigma(\nu'-1)$ and define $\tau_2$ to be equal to $\tau_1$ except that $\tau_2(\mu_1)=\tau_1(\mu_2)$ and $\tau_2(\mu_2)=\tau_1(\mu_1)$.  This gives 
\be{}
t_2=(-1)^{|\sigma_2|}z_{\sigma_2(\nu'-1)}^{\nu'-2}\partial_{\tau_2(\mu_1)}z_{\sigma_2(\nu')}^{\nu'-1}\times\dots,
\ee
where the hidden part is the same as in Eq. (\ref{lemma1t1}).  We also have that 

\be{}
z_{\sigma_1(\nu'-1)}^{\nu'-2}\partial_{\tau_1(\mu_1)}z_{\sigma_1(\nu')}^{\nu'-1} & =z_{\sigma_2(\nu'-1)}^{\nu'-2}\partial_{\tau_2(\mu_1)}z_{\sigma_2(\nu')}^{\nu'-1}\\
& =(\nu'-1)z_{\sigma_1(\nu'-1)}^{\nu'-2}z_{\sigma_1(\nu')}^{\nu'-2}
\ee
The permutations $\sigma_1$ and $\sigma_2$ have different parities, and we can therefore conclude that $t_1+t_2=0$ and that all terms of $\hat d_R \mathcal J$ are cancelled in this way.

\end{proof}

\begin{lemma}[Simple State Identity]
\label{lemma:SSidentity}
\be{}
P(\bm n)
=
(-1)^R
\hat d_R^{\sim \mu}
P(\bm n - R \bm e_\mu)
\ee
where
\be{}
\hat d_R^{\sim \mu}=
\N \sum_{\tau \in S_A} \prod_{\nu=1}^R \partial_{\tau(\nu)} (1-\delta_{\mu \tau(\nu)}).
\ee
\end{lemma}
In words, the lemma states that we can write a simple polynomial $P (\bm n)$ as the sum
of all simple polynomials where one, the $\mu$'th, occupation
number of $\bm n$ has been reduced by $R$ units and $R$ other
(not $\mu$) have been increased by one.

\begin{proof}
The proof is based on induction. 
We define the operator 

\be{}\hat d_R^{\mu} = 
\N \sum_{\tau \in S_A} \prod_{\nu=1}^R \partial_{\tau(\nu)} \delta_{\mu \tau(\nu)}
\ee 
which when acting on $P(\bm n)$ gives the sum of all simple polynomials where $R$ occupation numbers, always including $\mu$, have been increased by $1$.
By definition

\be{}
\hat d_R = \hat d_R^{\sim \mu} + \hat d_R^\mu.
\ee

Consider the polynomial $P(\bm n - 1 \bm e_\mu)$.
Lemma \ref{lemma:SymmetricTranslationalInvariance} proves the base case, since

\begin{multline}
0 
=\hat d_1 P(\bm n - 1 \bm e_\mu) \\
=
\hat d_1^\mu  P(\bm n - 1 \bm e_\mu)  + \hat d_1^{\sim \mu} P(\bm n - 1 \bm e_\mu) 
\\
 = P(\bm n) +\hat d_1^{\sim \mu} P(\bm n - 1 \bm e_\mu)
\end{multline}

\be{}
P(\bm n) = (-1)^1 \hat d_1^{\sim \mu} P(\bm n - 1 \bm e_\mu).
\ee

Now, we assume the Lemma holds for $R$, i.e.

\be{}
P(\bm n) = (-1)^R \hat d_R^{\sim \mu} P(\bm n - R \bm e_\mu).
\ee
Consider the polynomial $P(\bm n - (R+1) \bm e_\mu)$. 
Lemma \ref{lemma:SymmetricTranslationalInvariance} gives

\begin{multline}
0
=
\hat d_{R+1} P(\bm n - (R+1) \bm e_\mu )  \\
=
\hat d_{R+1}^{\mu} P(\bm n - (R+1) \bm e_\mu )  + \hat d_{R+1}^{\sim \mu} P(\bm n - (R+1) \bm e_\mu )\\
=
\hat d_{R}^{\sim \mu} P(\bm n - R \bm e_\mu )  + \hat d_{R+1}^{\sim \mu} P(\bm n - (R+1) \bm e_\mu ) 
\end{multline}
since the sum of simple polynomials where $R+1$ occupation in $\bm n$, always including $\mu$, have been increased by $1$ is the same as the sum of simple polynomials where $R$ occupation numbers in $(\bm n + \bm e_\mu)$, not including $\mu$, have been increased by $1$. 
By the induction hypothesis, this is

\be{}
(-1)^R P(\bm n) + \hat d_{R+1}^{\sim \mu} P(\bm n - (R+1) \bm e_\mu ) = 0
\ee

\be{}
 P(\bm n)  = (-1)^{R+1} \hat d_{R+1}^{\sim \mu} P(\bm n - (R+1) \bm e_\mu ).
\ee

So by the induction principle the Lemma holds for all $R\geq 1$.

\end{proof}

%%%%%%%%%%%%%%%%
%%% BIBLIOGRAPHY %%%%
%%%%%%%%%%%%%%%%

\end{document}